\documentclass{article}

\usepackage{enumerate}
\usepackage[a4paper]{geometry}
\usepackage{amssymb,amsmath,mathtools,xcolor,graphicx,xspace,colortbl,ragged2e,rotating} %
\usepackage{amsfonts}  %% 010
\usepackage{amsmath}  %% 010
\usepackage{amssymb}  %% 100
\usepackage{amsthm}  %% 100
\usepackage{color}  %% 100
\usepackage{graphicx}  %% 100
\usepackage{mathtools}  %% 100
\usepackage{hyperref}  %% 10000
\graphicspath{{ TanMei_graphics/}{TanMei_tcache/}{TanMei_gcache/}}
\DeclareGraphicsExtensions{.pdf,.eps,.ps,.png,.jpg,.jpeg}
\providecommand{\U}[1]{\protect\rule{.1in}{.1in}}
\allowdisplaybreaks
\newtheorem{theorem}{Theorem}

\newtheorem{corollary}[theorem]{Corollary}

\newtheorem{definition}[theorem]{Definition}
\newtheorem{example}[theorem]{Example}

\newtheorem{lemma}[theorem]{Lemma}
\newtheorem{notation}[theorem]{Notation}

\newtheorem{proposition}[theorem]{Proposition}
\newtheorem{remark}[theorem]{Remark}

\usepackage{cite}

\begin{document}

\def\sres{\operatorname*{sres}}

\title{B\'ezout Subresultants for Univariate Polynomials\\in General Basis}

\author{Jing Yang, Wei Yang\thanks{%
Corresponding author.}\\[5pt]
SMS--HCIC--School of Mathematics and Physics,\\
Center for Applied Mathematics of Guangxi,\\
Guangxi Minzu University, Nanning 530006, China\\
yangjing0930@gmail.com; weiyang020499@163.com}
\date{}
\maketitle
\begin{abstract}
{Subresultant is a powerful tool for developing various algorithms in computer algebra. Subresultants for polynomials in standard basis (i.e., power basis) have been well studied so far. With the popularity of basis-preserving algorithms, resultants and subresultants in non-standard basis are drawing more and more attention. In this paper, we develop a formula for B\'ezout subresultants of univariate polynomials in general basis, which covers a broad range of non-standard bases. More explicitly, the input polynomials are provided in a given general basis and the resulting subresultants are B\'ezout-type expressions in the same basis. It is shown that the  subresultants share the essential properties as the subresultants in standard basis.}

\end{abstract}

\section{Introduction}
Resultant theory plays a fundamental role in computer algebra.
Due to its importance,
extensive research has been carried out both in theoretical and practical aspects on resultants, subresultants and their variants (just list a few
\cite{sylvester1853,lascoux2003,diaz2004various,collins1967,
barnett1971greatest,Ho_1989,li2006,
terui2008,
bostan2017,
hy2021,
cox2021,
hongyang2021subresultant}).
However, most of the studies are focused on polynomials in standard basis
(also known as power basis). More explicitly, the input polynomials are formulated in standard basis  and so are the output subresultant polynomials. With the increasing popularity of basis-preserving algorithms in various applications \cite{fr1987,goodman1991,cp1993,aadr2006,jj2003}, people are more and more interested in resultants and subresultants for polynomials in non-standard basis (see \cite{bl2004,mm2007,acgs2007,wangyang2022sres}).
In this paper, we will study the B\'ezout subresultant polynomial in  general basis (see Definition \ref{def:general_basis}), which covers a wide range of non-standard bases.
Many well-known bases lie in this {category}, e.g., Newton basis and {Horner} basis, and standard basis can also be viewed as an instance of general basis. In the settings of the paper, the
input polynomials are expressed in  general basis and so are the output subresultant polynomials.

Typically subresultant polynomials are expressed in the form of determinental polynomial.
It is because determinental polynomials are the sum of minors of resultant matrices which possess nice algebraic structures. These structures  often bring lots of convenience for theoretical development and subsequent applications.
In this paper, we
follow this approach. Specifically, we  extend the concept of determinental polynomial from power basis to general basis and develop a generalized determinental polynomial formula for subresultant polynomials of two univariate polynomials.
The matrix used for formulating the determinental polynomial is B\'ezout matrix and thus we call the resulting subresultant polynomial \emph{B\'ezout subresultant polynomial}.
It is shown that the formulated subresultant polynomial possesses the  essential properties as the subresultant polynomial in standard basis (see Proposition \ref{thm:gcd}).

For developing the formula for B\'ezout subresultant polynomials of univariate polynomials in general basis, we first review the definition of B\'ezout matrix in non-standard basis given in \cite{acgs2007}.
Since  the degree of subresultant polynomial is often smaller than the input polynomials, in order to formulate the  subresultant polynomial in the given basis, we put a natural constraint on the non-standard basis and require the basis to be general. With  the above settings, we show that the subresultant polynomial can be written as the generalized determinental polynomial of a certain submatrix of the B\'ezout matrix in nonstandard basis.

Compared with other previous related works, the newly developed formula for the B\'ezout subresultant polynomial of univariate polynomials in general basis has the following features. First, it can be viewed as a generalization of the B\'ezout-type
subresultant polynomials in power basis \cite{houwang2000,apery2006resultant}. Second, it is also a generalization of the resultant in non-standard basis
\cite{barnett1987bezoutian,
hoon2002,
DazToca2014TheNS,
YANG2001,
Wu2010}.
It should be pointed out that the main contribution of the paper is that it verifies the
possibility of formulating subresultant polynomials in general basis
with B\'ezout matrix in theory rather than coming up with an efficient way to compute the subresultant in an arbitrarily given general basis.
Actually it is almost impossible to provide a uniform solution to the latter since it heavily relies on the structure of the chosen basis.

The paper is structured as follows. In Section \ref{sec:preliminaries}, we first review the concepts of B\'ezout resultant matrix and subresultant polynomial defined in roots  as well as  its inherent connection with the gcd problem.
Section \ref{sec:main_result} is devoted to present the  main result of the paper (see Theorem \ref{thm:main}).
The correctness of the main theorem is verified in Section \ref{sec:proofs}. The paper is concluded in Section \ref{sec:conclusion} with some further remarks.

\section{Preliminaries}\label{sec:preliminaries}

Throughout the paper, we assume $\mathbb{F}$  to be  the fractional field of an integral domain and $\overline{\mathbb{F}}$  the algebraic closure of $\mathbb{F}$. Let $\mathbb{F}_n[x]$ denote the set consisting of all polynomials in $\mathbb{F}[x]$ with degree no greater than $n$.

\subsection{B\'ezout resultant in general basis}
\begin{definition}\label{def:general_basis}
Let $\boldsymbol{\omega}(x)=(\omega_s,\ldots,\omega_{1},\omega_0)^{T}$ be the basis of $\mathbb{F}_s[x]$ where $\omega_i$ is monic and $\deg \omega_i=i$. Then we call $\boldsymbol{\omega}(x)$ (or $\boldsymbol{\omega}$ for short) a \emph{general basis} of $\mathbb{F}_s[x]$.
If no ambiguity occurs, we also call $\boldsymbol{\omega}(x)$ or $\boldsymbol{\omega}$ a {general basis} for short.
\end{definition}

For example, the standard basis $\boldsymbol{x}=(x^s,\ldots,x^{1},x^0)^{T}$ is a particular specialization of general basis of  $\mathbb{F}_s[x]$ since $\deg x^i=i$.

Another specialization of general basis frequently used is the Newton basis
$\boldsymbol{\omega}(x)=(\omega_s,\ldots,$
$\omega_{1},\omega_0)^{T}$ associated with $\lambda=(\lambda_s,\ldots,\lambda_1)\in\mathbb{F}^{s}$
where
$$
\omega_i=\left\{
\begin{array}{ll}
1, & \text{for~}i=0; \\
(x-\lambda_i)\omega_{i-1},  & \text{for~}i>0,
\end{array}
\right.
$$
since it is easy to verify that $\deg \omega_i=i$.
In this paper, we are mainly concerned about B\'ezout subresultants for polynomials in  general basis.

Consider $F,G\in\mathbb{F}[x]$ with degrees $n$ and $m$ (where $n>m$) in a general basis $\boldsymbol{\omega}(x)=(\omega_n,\ldots,\omega_1,\omega_0)^{T}$.
More explicitly,
we assume

\begin{equation}\label{eq:F+G}
F\left ( x \right ) =\sum_{i=0}^{n} a_{i} \omega_{i} ,\quad G\left ( x \right ) =\sum_{i=0}^{m} b_{i} \omega_{i } .
\end{equation}

To construct the B\'ezout matrix of $F$ and $G$ in $\boldsymbol{\omega}$, we recall the well known Cayley quotient below.
The Cayley quotient of $F$ and $G$ is defined as
\[C(x,y)=\frac{\left |\begin{matrix}
F\left ( x \right )  & F\left (  y\right ) \\
G\left ( x \right )  & G\left ( y \right )
\end{matrix} \right | }{x-y}.\]
It is noted that the Cayley quotient of two polynomials is independent on the basis used to formulate the input polynomials.

\begin{definition}[B\'ezout matrix]\label{def:bez_mat}
Let $\boldsymbol{\omega}=(\omega_n,$ $\ldots,\omega_1,\omega_0)^T$ be a general basis of $\mathbb{F}_n[x]$ and $F,G$ be as in \eqref{eq:F+G}.
Then
the \emph{B\'ezout matrix} of $F$ and $G$ in the basis $\boldsymbol{\omega}$ is defined as an $n\times n$ matrix $\boldsymbol{B_\omega}$ such that
\[C(x,y)=\bar{\boldsymbol{\omega}}(x)^T\cdot \boldsymbol{B_\omega}\cdot\bar{\boldsymbol{\omega}}(y)\]
where
$\bar{\boldsymbol{\omega}}=(\omega_{n-1},\ldots,\omega_0)^{T}$.
\end{definition}

Note that the B\'ezout matrix in Definition \ref{def:bez_mat} is highly dependent on the basis. When $\boldsymbol{\omega}=\boldsymbol{x}$, $\boldsymbol{B_\omega}$ becomes the familiar B\'ezout matrix in standard basis, denoted by $\boldsymbol{B}$.
Obviously, $\boldsymbol{B}=\boldsymbol{U}^T\boldsymbol{B_\omega}\boldsymbol{U}$ where $\boldsymbol{U}$ is the transition matrix from the basis $\bar{\boldsymbol{x}}=(x^{n-1},\ldots,x^1,x^0)^T$ to $\bar{\boldsymbol{\omega}}$.

\subsection{Subresultants in standard basis}
In classical resultant theory,  subresultant extends the concept of resultant and it is defined as the leading coefficient of the subresultant polynomial which is written as a determinental polynomial of a certain submatrix of the well known Sylvester resultant matrix of the input polynomials.
In \cite{houwang2000}, Hou and Wang proved that subresultant polynomial could also be expressed as the determinental polynomial of some submatrix of the B\'ezout matrix. It should be pointed out that the basis used to construct the B\'ezout matrix and formulate subresultant polynomials and subresultants in \cite{houwang2000} is standard basis.
Note that  resultant and subresultants in the previous discussion are formulated as polynomials in terms of the coefficients of $F$ and $G$ which highly depend on the basis and the adopted resultant matrices. In  \cite{1999_Hong,diaz2004various,2006D'Andrea}, Hong et al. provided  an equivalent definition of subresultant in terms of roots, which does not depend on both of them.
In the remaining part, we use the definition of subresultant in terms of roots. Before presenting the formal definition, we introduce the following notations for the sake of simplicity.

\begin{notation}\ 
\begin{itemize}
\item $\alpha_1,\ldots,\alpha_n$ are  the $n$ roots of $F$ over $\overline{\mathbb{F}}$;

\item $\boldsymbol{\alpha}:=(\alpha_1,\ldots,\alpha_n)$;

\item $\boldsymbol{\alpha}_{ij}:=(\alpha_i^{j},\ldots,\alpha_i^1,\alpha_i^0)^T$;

\item $\boldsymbol{V}(\boldsymbol{\alpha}):=\begin{bmatrix}
\alpha_1^{n-1}&\cdots&\alpha_n^{n-1}\\
\vdots&\ddots&\vdots\\
\alpha_1^{0}&\cdots&\alpha_n^{0}
\end{bmatrix}$;

\item $\boldsymbol{x}_j:=(x^j,\ldots,x^1,x^0)^T$.
\end{itemize}
\end{notation}

\begin{definition}\label{def:sres}
Given $F,G\in \mathbb{F}[x]$ with degree $n$ and $m$ respectively, let $a_n$ be the leading coefficients of $F$.
For $0\le k\le\min(m,n)$, the \emph{$k$-th subresultant polynomial} of $F$ and $G$ with respect to $x$ is defined as
\begin{equation}\label{eq:sres}
S_k:=\frac{c\left|
\begin{array}{lclc}
\boldsymbol{\alpha}_{1,n-k-1}G(\alpha_1)&\cdots&\boldsymbol{\alpha}_{n,n-k-1}G(\alpha_n)&\\
\boldsymbol{\alpha}_{1,k}&\cdots&\boldsymbol{\alpha}_{n,k}&\boldsymbol{x}_k
\end{array}\right|}{|\boldsymbol{V}(\boldsymbol{\alpha})|}
\end{equation}

where $c=(-1)^{k}a_n^{m-k}$. More explicitly,
\begin{equation*}
S_k=\frac{c\left|
\begin{array}{ccc|c}
\alpha_1^{n-k-1}G(\alpha_1) & \cdots & \alpha_n^{n-k-1}G(\alpha_n) &\\
\vdots & \ddots & \vdots &\\
\alpha_1^{0}G(\alpha_1) & \cdots & \alpha_n^{0}G(\alpha_n) &\\\hline
\alpha_1^k & \cdots & \alpha_n^k &x^k\\
\vdots & \ddots & \vdots &\vdots\\
\alpha_1^0 & \cdots & \alpha_n^0 &x^0
\end{array}\right|}{
\left|
\begin{array}{ccc}
\alpha_1^{n-1} & \cdots & \alpha_n^{n-1}\\
\vdots & \ddots & \vdots\\
\alpha_1^{0} & \cdots & \alpha_n^{0}
\end{array}
\right|}.
\end{equation*}

The coefficient of $S_k$ in $x^k$, denoted by $s_k$, is called the \emph{$k$-th subresultant} of $F$ and $G$ with respect to $x$.
\end{definition}

\begin{remark}\label{rem:interpreting}
The expression for $S_k$ in \eqref{eq:sres} should  be interpreted as follows, otherwise the denominator will vanish when $F$ is not squarefree.
\begin{enumerate}[(1)]
\item
Treat $\alpha_1,\ldots,\alpha_n$ as indeterminates and carry out the exact division, which results in a symmetric polynomial in terms of $\alpha_1,\ldots,,\alpha_{n}$.
\item
Evaluate the polynomial with $\alpha_1,\ldots,\alpha_n$  assigned the value of roots of $F$.
\end{enumerate}
\end{remark}

The following proposition captures the inherent connection between subresultant polynomials and the greatest common divisor of $F$ and $G$ and plays an essential role in many fundamental algorithms in computer algebra.

\begin{proposition}\label{thm:gcd}
Given $F,G\in \mathbb{F}[x]$, the following two conditions are equivalent:
\begin{enumerate}
\item $\deg \gcd(F,G)=k$;
\item $s_0=\cdots =s_{k-1}=0\wedge s_k\ne0$.
\end{enumerate}
Moreover, if $\deg \gcd(F,G)=k$, we have
\[\gcd(F,G)=S_k.\]
\end{proposition}

\section{Main Result}\label{sec:main_result}
To present the main theorem in this paper, we first generalize the determinental polynomial from the form in standard basis to that in general basis.

\begin{definition}[Generalized determinental polynomial]
Let $\boldsymbol{\omega}=(\omega_s,\ldots,\omega_1,\omega_0)^{T}$ be a {general basis} of $\mathbb{F}_s[x]$ and $\boldsymbol{M}\in\mathbb{F}^{(s-k)\times s}$ where $k<s$. Then the determinental polynomial of $\boldsymbol{M}$ associated with $\boldsymbol{\omega}$ is defined as
\[\operatorname*{detp}\nolimits_{\boldsymbol{\omega}}\boldsymbol{M}:=\sum_{i=0}^{k}|\hat{\boldsymbol{M}}_i|\cdot \omega_i\]
where $\hat{\boldsymbol{M}}_i$ is the submatrix of $\boldsymbol{M}$ consisting of the first $s-k-1$ columns and  the $(s-i)$-th column.
\end{definition}

\begin{theorem}[Main result]\label{thm:main}
Assume $\boldsymbol{\omega}=(\omega_n,\ldots,\omega_1,\omega_0)^{T}$ is a general basis of $\mathbb{F}_n[x]$. Let  $F$ and $G$ be as in \eqref{eq:F+G}
and $\boldsymbol{B}_{\boldsymbol{\omega},k}$
be the submatrix of the B\'ezout matrix of $F$ and $G$ in  $\boldsymbol{\omega}$ obtained by deleting the last $k$ rows. Then for $0\le k\le m$,
\begin{equation}\label{eq:Sk_in_thm}
S_k=c_{\boldsymbol{\omega}}\cdot\operatorname{detp}_{\boldsymbol{\omega}}\boldsymbol{B}_{\boldsymbol{\omega},k}
\end{equation}
where
$$c_{\boldsymbol{\omega}}=(-1)^{n-k\choose 2}a_n^{m-n}.$$ \end{theorem}

The determinental polynomial given by \eqref{eq:Sk_in_thm}
provides an equivalent expression for $S_k$ and thus the determinental polynomial herein shares Proposition \ref{thm:gcd}.

\begin{example}\label{ex}
Let
$\boldsymbol{\nu}=(\nu_{3},\nu_2,\nu_1,\nu_0)^T$
be the Newton basis associated with $\lambda=(1,0,2)$, i.e.,
\[
\begin{array}{llll}
\nu_0=1,&\nu_1=x-2,&\nu_2=x(x-2),&
\nu_3=x(x-1)(x-2).
\end{array}
\]
Given
\begin{align*}
F\left ( x \right ) &=a_{3} \nu_{3} +a_{2} \nu_{2} +a_{1} \nu_1+a_{0}\nu_0,\\
G\left ( x \right ) &= b_{2} \nu_{2} +b_{1}\nu_1+b_{0}\nu_0,
\end{align*}
where $a_3b_2\ne0$, we compute the first subresultant $S_1$ of $F$ and $G$ with respect to $x$. It is required that $S_1$ is expressed in the basis $\boldsymbol{\nu}$.

By Definition \ref{def:bez_mat}, the B\'ezout matrix of $F$ and $G$ in  $\boldsymbol{\nu}$ is
\[\boldsymbol{B}_{\boldsymbol{\nu}}=
\begin{bmatrix}
a_{3}b_{2}   & a_{3}b_{1}  & a_{3}b_{0} \\
a_{3}b_{1}& P_1 &P_2 \\
a_{3}b_{0} & P_2 &P_3
\end{bmatrix}\]
where
\begin{align*}
P_1&=a_{3}b_{0}-a_{3}b_{1}+a_{2}b_{1}-a_{1}b_{2},\\
P_2&=a_{2}b_{0}-a_{0}b_{2}+a_{3}b_{0},\\
P_3&=a_{1}b_{0}+2a_{2}b_{0}+2a_{3}b_{0}-a_{0}b_{1}-2a_{0}b_{2}.
\end{align*}
By Theorem \ref{thm:main},
the first subresultant of $F$ and $G$ in $\boldsymbol{\nu}$ is
\begin{equation}\label{ex:S1_detp}
S_1=c_{\boldsymbol{\nu}}\cdot\left(
\begin{vmatrix}
a_{3}b_{2}   & a_{3}b_{1} \\
a_{3}b_{1}& P_1
\end{vmatrix}\nu_1+
\begin{vmatrix}
a_{3}b_{2}   & a_{3}b_{0} \\
a_{3}b_{1}&P_2
\end{vmatrix}\nu_0\right)
\end{equation}
where
\[c_{\boldsymbol{\nu}}=(-1)^{3-1\choose 2}a_3^{  2-3}=-a_3^{  -1}.\]
The expansion of \eqref{ex:S1_detp} yields
\begin{align}
S_1=&\,\left( a_{{1}}b_{{2}}^{2}-a_{{2}}b_{{1}}b_{{2}}-a_{{3}}b_{{0
}}b_{{2}}+a_{{3}}b_{{1}}^{2}+a_{{3}}b_{{1}}b_{{2}} \right)
\nu_1+\left( a_{{0}}b_{{2}}^{2}-a_{{2}}b_{{0}}b_{{2}}+a_{{3}}b_{{0
}}b_{{1}}-a_{{3}}b_{{0}}b_{{2}} \right) \nu_0.\label{ex:S1_nu}
\end{align}
Next we will verify that
$S_1$ is the same as the first subresultant of $F$ and $G$
when converted into expressions in the standard basis.

Converting $F$ and $G$ into expressions in the standard basis $\boldsymbol{x}=(x^3,x^2,$ $x^1,x^0)^T$, we have
\begin{align*}
F&=a_{3}x^3 + \left ( a_{2} - 3a_{3} \right )x^2  + \left ( a_{1} - 2a_{2} + 2a_{3} \right )x  + \left (a_{0} - 2a_{1}\right),\\
G&=b_{2}x^2 + \left ( b_{1} - 2b_{2} \right )x  + b_{0} - 2b_{1}.
\end{align*}
Then the B\'ezout matrix of $F$ and $G$ with respect to $x$ in the basis $\boldsymbol{x}=(x^3,x^2,x^1,x^0)^{T}$ is
\[\boldsymbol{B}=\begin{bmatrix}
a_{3}b_{2}   & a_{3}b_{1}-2a_{3}b_{2}  & a_{3}b_{0}-2a_{3}b_{1} \\
a_{3}b_{1}-2a_{3}b_{2}& P_4 &P_{5} \\
a_{3}b_{0}-2a_{3}b_{1} & P_5 &P_6
\end{bmatrix} \]
where
\begin{align*}
P_4=&\,(a_{{2}}b_{{1}}-a_{{1}}b_{{2}})+a_{{3}}b_{{0}}-5\,a_{{3}}b_{{1}}+4\,a_{{3}}b_{{2}},\\
P_5=&\,(a_{2}b_{0}-a_{0}b_{2})+2(a_{1}b_{2}-a_{2}b_{1})-3a_{3}b_{0}+6a_{3}b_{1},\\
P_6=&\,(a_{1}b_{0}-a_{0}b_{1})+2(a_{0}b_{2}-a_{2}b_{0})+4(a_{2}b_{1}-a_{1}b_{2})+2a_{3}b_{0}-4a_{3}b_{1}.
\end{align*}
Thus
\begin{align}
S_1=&\,c_{\boldsymbol{x}}\left(
\begin{vmatrix}
a_{3}b_{2}   & a_{3}b_{1}-2a_{3}b_{2} \\
a_{3}b_{1}-2a_{3}b_{2}& P_4
\end{vmatrix}x^1+\begin{vmatrix}
a_{3}b_{2}   & a_{3}b_{0}-2a_{3}b_1 \\
a_{3}b_{1}-2a_{3}b_{2}&P_5
\end{vmatrix}x^0\right)\label{ex:S1_sb_detp}
\end{align}
where again $c_{\boldsymbol{x}}=  -a_3^{-1}$. The expansion of \eqref{ex:S1_sb_detp} yields
\begin{align*}
S_1=\,&\left( a_{{1}}b_{{2}}^{2}-a_{{2}}b_{{1}}b_{{2}}-a_{{3}}b_{{0
}}b_{{2}}+a_{{3}}b_{{1}}^{2}+a_{{3}}b_{{1}}b_{{2}} \right)
x+\\
& \left( a_{{0}}b_{{2}}^{2}-2\,a_{{1}}b_{{2}}^{2}-a_{{2}}b_
{{0}}b_{{2}}+2\,a_{{2}}b_{{1}}b_{{2}}+a_{{3}}b_{{0}}b_{{1}}+a_{{3}}b_{
{0}}b_{{2}}-2\,a_{{3}}b_{{1}}^{2}-2\,a_{{3}}b_{{1}}b_{{2}} \right),
\end{align*}
which is exactly $S_1$ in \eqref{ex:S1_nu} when reformulated in standard basis.
\end{example}

\section{Proof of the Main Result}\label{sec:proofs}

Before going into the details of the proof, we give a brief sketch of  the proof.
In the first stage, we convert $S_k$  in terms of roots from power basis to general basis.  Then $S_k$ in general basis is converted from a polynomial in roots  to a determinant in coefficients. Finally, we show that the determinant in coefficients can be written as a generalized determinental polynomial given by Theorem \ref{thm:main}.

\subsection{Converting $S_k$ in roots from standard basis to general basis}

It is obvious that the subresultant $S_k$ of $F$ and $G$ with respect to $x$ given by \eqref{eq:sres} does not depend on the basis that $F$ and $G$ are expressed with. However, with a closer look at the expression in  \eqref{eq:sres}, one may notice that the standard basis actually appears in $S_k$.
More explicitly, the entries of determinants in $S_k$ are of the form $x^i~(0\le i\le k)$ or the evaluation of $x^i~(0\le i\le k)$ and $x^ig(x)~(0\le i\le n-k-1)$ at the roots of $F$. In this sense, we can say that $S_k$ is expressed in standard basis. In this subsection, we will convert $S_k$ in terms of roots from standard basis to general basis.
For simplicity, we introduce the following  notations.

\begin{notation}\ 
\begin{itemize}
\item $\boldsymbol{\omega}(x):=(\omega_n,\ldots,\omega_{1},\omega_0)^T$ is  a general basis of $\mathbb{F}_n[x];$

\item $\boldsymbol{\omega}_j(x):=(\omega_j,\ldots,\omega_1,\omega_0)^T$ for $0\le j<n$;

\item $\boldsymbol{W}(\boldsymbol{\alpha}):=\begin{bmatrix}
\omega_{n-1}(\alpha_1)&\cdots&\omega_{n-1}(\alpha_n)\\
\vdots&\ddots&\vdots\\
\omega_{0}(\alpha_1)&\cdots&\omega_{0}(\alpha_n)
\end{bmatrix}$.
\end{itemize}
\end{notation}

\begin{lemma}\label{lem:std2nstd}
Let  $F,G\in\mathbb{F}[x]$ be such that
$\deg F=n$. Then for $k\le n$, we have
\begin{equation}\label{eq:sres_nstd_basis}
S_k=\frac{{c}\cdot\left|
\begin{array}{lclc}
\boldsymbol{\omega}_{n-k-1}(\alpha_1)G(\alpha_1) & \cdots & \boldsymbol{\omega}_{n-k-1}(\alpha_n)G(\alpha_n) &\\
\boldsymbol{\omega}_{k}(\alpha_1) & \cdots & \boldsymbol{\omega}_{k}(\alpha_n)&\boldsymbol{\omega}_{k}
\end{array}\right|}{
|\boldsymbol{W}(\boldsymbol{\alpha})|}
\end{equation}
where
$c$ is as in Definition \ref{def:sres}, i.e, $c=(-1)^ka_n^{m-k}$.
% More explicitly,
% \begin{equation*}
% S_k=\frac{c_{1}\cdot\left|
% \setlength{\arraycolsep}{1.0pt}
% \begin{array}{ccc|c}
% \omega_{n-k-1}(\alpha_1)G(\alpha_1) & \cdots & \omega_{n-k-1}(\alpha_n)G(\alpha_n) &\\
% \vdots & \ddots & \vdots &\\
% \omega_0(\alpha_1)G(\alpha_1) & \cdots & \omega_0(\alpha_n)G(\alpha_n) &\\\hline
% \omega_k(\alpha_1) & \cdots & \omega_k(\alpha_n) &\omega_k(x)\\
% \vdots & \ddots & \vdots &\vdots\\
% \omega_0(\alpha_1) & \cdots & \omega_0(\alpha_n) &\omega_0(x)
% \end{array}\right|}{
% \left|
% \begin{array}{ccc}
% \omega_{n-1}(\alpha_1) & \cdots & \omega_{n-1}(\alpha_n)\\
% \vdots & \ddots & \vdots\\
% \omega_{0}(\alpha_1) & \cdots & \omega_{0}(\alpha_n)
% \end{array}
% \right|}.
%\end{equation*}
\end{lemma}

\begin{proof}
We will prove the lemma by comparing the denominators  and numerators in  \eqref{eq:sres} and \eqref{eq:sres_nstd_basis} and showing their equivalence, respectively.
The proof is a bit long and thus we divide it into several steps.

\smallskip
\begin{enumerate}[(i)]
\item  Show the equivalence of denominators in  \eqref{eq:sres} and \eqref{eq:sres_nstd_basis}.
\begin{enumerate}[(a)]
\item Let $\boldsymbol{U}$ be the transition matrix from $\bar{\boldsymbol{x}}=(x^{n-1},\ldots,x^0)^{T}$ to $\bar{\boldsymbol{\omega}}=(\omega_{n-1},\ldots,\omega_0)^{T}$, i.e., $\bar{\boldsymbol{x}}=\boldsymbol{U\bar{\omega}}$.

Thus for $i=1,\ldots,n$,
$\bar{\boldsymbol{x}}(\alpha_i)=\boldsymbol{U}\bar{\boldsymbol{\omega}}(\alpha_i)$. More explicitly, we have
\[\boldsymbol{U}=\begin{bmatrix}
1&\cdots&\cdot\\[-3pt]
&\ddots&\vdots\\[-3pt]
&&1
\end{bmatrix}.\]
In other words, $U$ is a unit upper triangular matrix.

\item It follows that
\begin{align*}
\boldsymbol{V}(\boldsymbol{\alpha})=&
\begin{bmatrix}
\bar{\boldsymbol{x}}(\alpha_1)&\cdots&\bar{\boldsymbol{x}}(\alpha_n)
\end{bmatrix}\\
=&\begin{bmatrix}
\boldsymbol{U}\bar{\boldsymbol{\omega}}(\alpha_1)&\cdots&\boldsymbol{U}\bar{\boldsymbol{\omega}}(\alpha_n)
\end{bmatrix}\\
=&\,\boldsymbol{U}\begin{bmatrix}
\bar{\boldsymbol{\omega}}(\alpha_1)&\cdots&\bar{\boldsymbol{\omega}}(\alpha_n)
\end{bmatrix}\\
=&\,\boldsymbol{U}\boldsymbol{W}(\boldsymbol{\alpha})
\end{align*}

\item Taking determinants on both sides of the above equation, we have
$
|\boldsymbol{V}(\boldsymbol{\alpha})|=|\boldsymbol{U}|\cdot |\boldsymbol{W}(\boldsymbol{\alpha})|.
$

\item Since $\boldsymbol{U}$ is unit upper triangular, $|\boldsymbol{U}|=1$, which implies that
\begin{equation}\label{eq:V_W}
|\boldsymbol{V}(\boldsymbol{\alpha})|= |\boldsymbol{W}(\boldsymbol{\alpha})|.
\end{equation}
\end{enumerate}

\item  Show the equivalence of numerators in   \eqref{eq:sres} and \eqref{eq:sres_nstd_basis}.
\begin{enumerate}[(a)]
\item  Let $\boldsymbol{U}_j$ be the submatrix of $\boldsymbol{U}$ obtained by selecting its last $j$ rows and the last $j$ columns. Then
$\boldsymbol{U}_j$ is the transition matrix from the basis $\boldsymbol{x}_{j-1}$ of $\mathbb{F}_{j-1}[x]$ to $\boldsymbol{\omega}_{j-1}(x)$, i.e.,
$\boldsymbol{x}_{j-1}(x)=\boldsymbol{U}_{j}\boldsymbol{\omega}_{j-1}(x)$.

\item  It follows that
\begin{align*}
&\,\left[
\begin{array}{cccc}
\boldsymbol{x}_k(\alpha_1) & \cdots & \boldsymbol{x}_k(\alpha_n) &\boldsymbol{x}_k
\end{array}\right]\\
=&\,\left[
\begin{array}{cccc}
\boldsymbol{U}_{k+1}\boldsymbol{\omega}_{k}(\alpha_1) & \cdots & \boldsymbol{U}_{k+1}\boldsymbol{\omega}_{k}(\alpha_n) &\boldsymbol{U}_{k+1}\boldsymbol{\omega}_{k}(x)
\end{array}
\right]\\
=&\,\boldsymbol{U}_{k+1}
\left[
\begin{array}{cccc}
\boldsymbol{\omega}_{k}(\alpha_1) & \cdots & \boldsymbol{\omega}_{k}(\alpha_n) &\boldsymbol{\omega}_{k}(x)
\end{array}
\right].
\end{align*}

\item  With the same manner, we derive the following:
\begin{align*}
&
\left[
\begin{array}{cccc}
\boldsymbol{\alpha}_{1,n-k-1}G(\alpha_1) & \cdots & \boldsymbol{\alpha}_{n,n-k-1}G(\alpha_n)
\end{array}
\right]\\
=\,&\boldsymbol{U}_{n-k}
\left[
\begin{array}{cccc}
\boldsymbol{\omega}_{n-k-1}(\alpha_1)G(\alpha_1) & \cdots & \boldsymbol{\omega}_{n-k-1}(\alpha_n)G(\alpha_n) &\boldsymbol{0}_{(n-k)\times 1}\end{array}
\right].
\end{align*}

\item  Assembling the obtained expressions in (ii.b) and (ii.c), we obtain
\begin{align*}
&\left[
\begin{array}{lclc}
\boldsymbol{\alpha}_{1,n-k-1}G(\alpha_1)&\cdots&\boldsymbol{\alpha}_{n,n-k-1}G(\alpha_n)&\\
\boldsymbol{\alpha}_{1,k}&\cdots&\boldsymbol{\alpha}_{n,k}&\boldsymbol{x}_k
\end{array}
\right]\\[3pt]
=&
\left[
\begin{array}{l}
\boldsymbol{U}_{n-k}\!\!\left[
\begin{array}{lcll}
\boldsymbol{\omega}_{n-k-1}(\alpha_1)G(\alpha_1) & \cdots & \boldsymbol{\omega}_{n-k-1}(\alpha_n)G(\alpha_n) &\hspace{3.3em}
\end{array}\right]\\[5pt]
\boldsymbol{U}_{k+1}\!
\left[
\begin{array}{lcll}
\boldsymbol{~\omega}_{k}(\alpha_1)~ & \hspace{4.7em}\cdots & \boldsymbol{\omega}_{k}(\alpha_n) &\hspace{5.em}\boldsymbol{\omega}_{k}(x)
\end{array}\right]
\end{array}\right]\\[3pt]
=&
\begin{bmatrix}
\boldsymbol{U}_{n-k}&\\
&\boldsymbol{U}_{k+1}
\end{bmatrix}\cdot
\left[\begin{array}{lcll}
\boldsymbol{\omega}_{n-k-1}(\alpha_1)G(\alpha_1) & \cdots & \boldsymbol{\omega}_{n-k-1}(\alpha_n)G(\alpha_n) &\hspace{1.7em}\\
\boldsymbol{~\omega}_{k}(\alpha_1)~ & \cdots & \boldsymbol{\omega}_{k}(\alpha_n) &\boldsymbol{\omega}_{k}(x)
\end{array}\right].
\end{align*}

\item  Taking determinants on both sides yields
\begin{align*}
&\left|
\begin{array}{lcll}\boldsymbol{\alpha}_{1,n-k-1}G(\alpha_1)&\cdots&\boldsymbol{\alpha}_{n,n-k-1}G(\alpha_n)&\\
\boldsymbol{\alpha}_{1,k}&\cdots&\boldsymbol{\alpha}_{n,k}&\boldsymbol{x}_k
\end{array}\right|\\[3pt]
=&
\begin{vmatrix}
\boldsymbol{U}_{n-k}&\\
&\boldsymbol{U}_{k+1}
\end{vmatrix}\cdot
\left|
\begin{array}{lcll}
\boldsymbol{\omega}_{n-k-1}(\alpha_1)G(\alpha_1) & \cdots & \boldsymbol{\omega}_{n-k-1}(\alpha_n)G(\alpha_n) &\hspace{1.7em}\\
\boldsymbol{~\omega}_{k}(\alpha_1)~ & \cdots & \boldsymbol{\omega}_{k}(\alpha_n) &\boldsymbol{\omega}_{k}(x)
\end{array}\right|
.
\end{align*}

\item Since $\boldsymbol{U}$ is unit upper triangular, so are $\boldsymbol{U}_{n-k}$ and $\boldsymbol{U}_{k+1}$. Hence $|\boldsymbol{U}_{n-k}|=|\boldsymbol{U}_{k+1}|=1$, which  implies that
\begin{align}
&\left|
\begin{array}{lcll}\boldsymbol{\alpha}_{1,n-k-1}G(\alpha_1)&\cdots&\boldsymbol{\alpha}_{n,n-k-1}G(\alpha_n)&\\
\boldsymbol{\alpha}_{1,k}&\cdots&\boldsymbol{\alpha}_{n,k}&\boldsymbol{x}_k
\end{array}\right|\notag\\[3pt]
=&
\left|
\begin{array}{lcll}
\boldsymbol{\omega}_{n-k-1}(\alpha_1)G(\alpha_1) & \cdots & \boldsymbol{\omega}_{n-k-1}(\alpha_n)G(\alpha_n) &\hspace{1.7em}\\
\boldsymbol{~\omega}_{k}(\alpha_1)~ & \cdots & \boldsymbol{\omega}_{k}(\alpha_n) &\boldsymbol{\omega}_{k}(x)
\end{array}\right|
.\label{eq:numer_equiv}
\end{align}
\end{enumerate}

\item  Combining \eqref{eq:V_W} and \eqref{eq:numer_equiv}, we obtain
\begin{align*}
S_k=&\frac{c\cdot \left|\begin{array}{lcll}
\boldsymbol{\omega}_{n-k-1}(\alpha_1)G(\alpha_1) & \cdots & \boldsymbol{\omega}_{n-k-1}(\alpha_n)G(\alpha_n) &\hspace{1.7em}\\
\boldsymbol{~\omega}_{k}(\alpha_1)~ & \cdots & \boldsymbol{\omega}_{k}(\alpha_n) &\boldsymbol{\omega}_{k}(x)
\end{array}\right|}{|\boldsymbol{W}(\boldsymbol{\alpha})|}
\end{align*}
where $c=(-1)^ka_n^{m-k}$.
\end{enumerate}
\end{proof}

\begin{remark}

The expression for $S_k$ in \eqref{eq:sres_nstd_basis} should  be interpreted in the same manner as stated in Remark \ref{rem:interpreting}, otherwise the denominator will vanish when $F$ is not squarefree.
\end{remark}

\subsection{Converting $S_k$ in general basis from an expression in roots to that in coefficients}

The following Lemma \ref{lem:root2coeffs} allows us to convert $S_k$ in roots to an expression in coefficients of the input polynomials formulated with general basis. It generalizes the ides introduced by Li for subresultant polynomials in standard basis \cite{li2006}.

\begin{lemma}\label{lem:root2coeffs}
Given a general basis $\boldsymbol{\omega}=(\omega_{n},\ldots,\omega_1,\omega_0)^T$ of $\mathbb{F}_n[x]$ and $F$, $G$ as in \eqref{eq:F+G},
let $\boldsymbol{B_\omega}$ be the B\'ezout matrix of $F$ and $G$ with respect to $x$ in  $\boldsymbol{\omega}$ and $\boldsymbol{B}_{\boldsymbol{\omega},k}$ be the submatrix of  $\boldsymbol{B_\omega}$ by deleting the last $k$ rows. Then
$$S_k=c_{\boldsymbol{\omega}}\cdot
\begin{vmatrix}
\boldsymbol{B}_{\boldsymbol{\omega},k}\\
\boldsymbol{X}_{\boldsymbol{\omega},k}
\end{vmatrix}$$
where
\begin{itemize}
\item
$c_{\boldsymbol{\omega}}=(-1)^{n-k\choose 2}a_n^{m-n}$,

\item and \begin{equation}\label{eq:Xw}
\boldsymbol{X}_{\boldsymbol{\omega},k}
=\begin{bmatrix}
&&-1&&&\omega_k/\omega_0\\
&&&\ddots&&\vdots\\
&&&&-1&\omega_1/\omega_0
\end{bmatrix}_{k\times n}.
\end{equation}
\end{itemize}
\end{lemma}

Before  proving Lemma \ref{lem:root2coeffs}, we need to verify the following Lemmas \ref{lem:upperpart} and \ref{lem:lowerpart} first, which captures the essential ingredients of Lemma \ref{lem:root2coeffs}.

%         \begin{definition}[Elementary symmetric polynomial]
%                 \end{definition}
\begin{lemma}\label{lem:upperpart}
Given $F$, $G$ as in \eqref{eq:F+G} with  $\boldsymbol{\omega}=(\omega_{n},\ldots,\omega_{1},\omega_0)^{T}$ to be a general basis  of $\mathbb{F}_n[x]$,
let
\begin{itemize}
\item
$\boldsymbol{B_\omega}$ be the B\'ezout matrix of $F$ and $G$ with respect to $x$ in $\boldsymbol{\omega}$,

\item $\bar{\boldsymbol{\omega}}=(\omega_{n-1},\ldots,\omega_1,\omega_0)^{T}$, and
\item
$\alpha_1,\ldots,\alpha_n$ be the $n$ roots of $F$ over $\overline{\mathbb{F}}$.
\end{itemize}
Then
\begin{align*}
\boldsymbol{B_\omega}\cdot
\begin{bmatrix}
\bar{\boldsymbol{\omega}}(\alpha_1)&\cdots&\bar{\boldsymbol{\omega}}(\alpha_n)
\end{bmatrix}
=a_{n}\boldsymbol{U}^T\boldsymbol{TU}
\begin{bmatrix}
\bar{\boldsymbol{\omega}}(\alpha_1)G(\alpha_1)&\cdots&\bar{\boldsymbol{\omega}}(\alpha_n)G(\alpha_n)
\end{bmatrix}
\end{align*}
where $\boldsymbol{U}$ is the transition matrix from $\bar{\boldsymbol{x}}$ to $\boldsymbol{\bar{\omega}}$ and
\[
\boldsymbol{T}=\begin{bmatrix}
&&&(-1)^{0}e_0\\
&&(-1)^{2}e_0&(-1)^{1}e_1\\
&\begin{sideways}
$\ddots$
\end{sideways}&\vdots&\vdots\\
(-1)^{2(n-1)}e_0&\cdots&\cdots&(-1)^{n-1}e_{n-1}
\end{bmatrix}.
\]
\end{lemma}

\begin{proof}
The proof is long and will be divided into several steps.
\begin{enumerate}[(i)]
\item Let
\begin{equation}\label{eq:cayley_quotient}
C(x,y)=\frac{\left |\begin{matrix}
F\left ( x \right )  & F\left (  y\right ) \\
G\left ( x \right )  & G\left ( y \right )
\end{matrix} \right | }{x-y}.
\end{equation}
By Definition \ref{def:bez_mat}, $\boldsymbol{B_\omega}$ satisfies
\begin{equation}\label{eq:Cayley_bez}
C(x,y)=\bar{\boldsymbol{\omega}}(x)^T\cdot \boldsymbol{B_\omega}\cdot\bar{\boldsymbol{\omega}}(y)
\end{equation}

\item Since $F(\alpha_i)=0$ for $i=1,\ldots,n$,
by setting $y=\alpha_i$ in \eqref{eq:Cayley_bez}, we get
\begin{equation}\label{eq:C_alphai_1}
C(x,\alpha_i)=\boldsymbol{\bar{\boldsymbol{\omega}}}(x)^T\cdot \boldsymbol{B_\omega}\cdot\bar{\boldsymbol{\omega}}(\alpha_i).
\end{equation}

\item On the other hand, from \eqref{eq:cayley_quotient}, we have
\[
C \left ( x,\alpha _{i} \right )=\frac{F\left ( x \right )G\left ( \alpha _{i} \right )  }{x- \alpha _{i} }
=a_{n}\prod_{\substack{j=1\\ j\ne i}}^{n}\left ( x-\alpha _{j} \right ) G\left ( \alpha _{i} \right ).
\]

\item Let $e_{j}^{\left ( i \right ) }  $ denote the $j$-th elementary symmetric polynomial\footnote{The $k$-th elementary symmetric polynomial in $x_1,x_2,\ldots,x_n$ is defined as
$\sigma_k=\sum\limits_{\substack{1\le i_1<\cdots<i_k\le n\\
\{i_1,\ldots,i_k\}\subset\{1,\ldots,n\}}}x_{i_1}\cdots x_{i_k}$.} in $\alpha_{1} ,\dots ,\alpha_{i-1} ,\alpha_{i+1},\dots, \alpha_{n}$ (with $e_{0}^{\left ( i \right ) } :=1$). Then
\begin{align*}
C \left ( x,\alpha _{i} \right )
&=a_{n} \left ( \sum_{j=0}^{n-1}(-1)^je_{j}^{(i)}x^{n-1-j}\right )G(\alpha_i)\\
&=a_{n}
\begin{bmatrix}
x^{n-1}&\cdots&x^{0}
\end{bmatrix}
\begin{bmatrix}
(-1)^{0}e_{0}^{(i)}\\
\vdots\\
(-1)^{n-1}e_{n-1}^{(i)}
\end{bmatrix}G(\alpha_i)\\
&=a_{n}
\bar{\boldsymbol{x}}^T
\begin{bmatrix}
(-1)^{0}e_{0}^{(i)}\\
\vdots\\
(-1)^{n-1}e_{n-1}^{(i)}
\end{bmatrix}G(\alpha_i).
\end{align*}

\item Since $\bar{\boldsymbol{x}}=\boldsymbol{U\bar{\omega}}$, we have
\begin{align}
C \left ( x,\alpha _{i} \right )
=a_{n}\boldsymbol{\bar{\omega}}^T\boldsymbol{U}^T
\begin{bmatrix}
(-1)^{0}e_{0}^{(i)}\\
\vdots\\
(-1)^{n-1}e_{n-1}^{(i)}
\end{bmatrix}G(\alpha_i)
=a_{n}\boldsymbol{\bar{\omega}}^T\boldsymbol{U}^{T}
\begin{bmatrix}
(-1)^{0}e_{0}^{(i)}G(\alpha_i)\\
\vdots
\\(-1)^{n-1}e_{n-1}^{(i)}G(\alpha_i)
\end{bmatrix}.\label{eq:C_alphai_2}
\end{align}

\item Comparing the coefficients of $C \left ( x,\alpha _{i} \right )$ in \eqref{eq:C_alphai_1} and \eqref{eq:C_alphai_2}, we obtain
\[\boldsymbol{B_\omega}\bar{\boldsymbol{\omega}}(\alpha_i)=a_{n}\boldsymbol{U}^T\begin{bmatrix}
e_{0}^{(i)}G(\alpha_i)\\
\vdots\\
(-1)^{n-1}e_{n-1}^{(i)}G(\alpha_i)
\end{bmatrix}.\]

\item
Assembling $\boldsymbol{B_\omega}\bar{\boldsymbol{\omega}}(\alpha_i)$ for $i=1,\ldots,n$ horizontally into a single matrix, we get
\begin{align}
&
\begin{bmatrix}
\boldsymbol{B_\omega}\bar{\boldsymbol{\omega}}(\alpha_1)&\cdots&\boldsymbol{B_\omega}\bar{\boldsymbol{\omega}}(\alpha_n)
\end{bmatrix}
\notag\\[5pt]
=\,&
\begin{bmatrix}
a_{n}\boldsymbol{U}^T\begin{bmatrix}
e_{0}^{(1)}G(\alpha_1)\\
\vdots\\
(-1)^{n-1}e_{n-1}^{(1)}G(\alpha_1)
\end{bmatrix}&\cdots&a_{n}\boldsymbol{U}^T\begin{bmatrix}
e_{0}^{(n)}G(\alpha_n)\\
\vdots\\
(-1)^{n-1}e_{n-1}^{(n)}G(\alpha_n)
\end{bmatrix}
\end{bmatrix}\notag\\[5pt]
=\,&a_{n}\boldsymbol{U}^T
\begin{bmatrix}
(-1)^{0}e_{0}^{(1)}G(\alpha_1)&\cdots&(-1)^{0}e_{0}^{(n)}G(\alpha_n)\\
\vdots&\ddots&\vdots\\
(-1)^{n-1}e_{n-1}^{(1)}G(\alpha_1)&\cdots&(-1)^{n-1}e_{n-1}^{(n)}G(\alpha_n) \end{bmatrix}\notag\\[5pt]
=&a_{n}\boldsymbol{U}^T
\begin{bmatrix}
(-1)^{0}e_{0}^{(1)}&\cdots&(-1)^{0}e_{0}^{(n)}\\
\vdots&\ddots&\vdots\\
(-1)^{n-1}e_{n-1}^{(1)}&\cdots&(-1)^{n-1}e_{n-1}^{(n)}
\end{bmatrix}
\begin{bmatrix}
G(\alpha_1)&&\\
&\ddots&\\
&&G(\alpha_n)
\end{bmatrix}\notag\\
=&a_{n}\boldsymbol{U}^T
\begin{bmatrix}
(-1)^{0}&&\\
&\ddots&\\
&&(-1)^{n-1}
\end{bmatrix}
\begin{bmatrix}
e_{0}^{(1)}&\cdots&e_{0}^{(n)}\\
\vdots&\ddots&\vdots\\
e_{n-1}^{(1)}&\cdots&e_{n-1}^{(n)}
\end{bmatrix}
\begin{bmatrix}
G(\alpha_1)&&\\
&\ddots&\\
&&G(\alpha_n)
\end{bmatrix}.\label{eq:expansion_ei}
\end{align}

\item Recall \cite[Lemma 35]{hongyang2021subresultant} which states that
\[
e_{j}^{\left ( i \right ) }=\sum_{k=0}^{j}(-1)^ke_{j-k}\alpha_i^k=
\setlength{\arraycolsep}{1.6pt}
\begin{bmatrix}
0&\cdots&0&(-1)^je_0&\cdots&(-1)^0e_j
\end{bmatrix}\bar{\boldsymbol{x}}(\alpha_i)
% \begin{bmatrix}
% \alpha_i^{n-1}\\
% \vdots\\
% \alpha_i^{j}\\
% \vdots\\
% \alpha_i^{0}
% \end{bmatrix}
\]
where $e_{j}$ denotes the $j$-th elementary symmetric polynomial in $\alpha_{1} ,\dots ,\alpha_{n}$ (with $e_{0}:=1$).
Therefore,
\begin{align}
&
\setlength{\arraycolsep}{2.5pt}
\begin{bmatrix}
(-1)^{0}&&\\
&\ddots&\\
&&(-1)^{n-1}
\end{bmatrix}
\begin{bmatrix}
e_{0}^{(1)}&\cdots&e_{0}^{(n)}\\
\vdots&\ddots&\vdots\\
e_{n-1}^{(1)}&\cdots&e_{n-1}^{(n)}
\end{bmatrix}\notag\\
=&
\setlength{\arraycolsep}{0.5pt}
\begin{bmatrix}
(-1)^{0}&&\\
&\ddots&\\
&&(-1)^{n-1}
\end{bmatrix}
\begin{bmatrix}
&&&(-1)^{0}e_0\\
&&(-1)^{1}e_0&(-1)^{0}e_1\\
&\begin{sideways}
$\ddots$
\end{sideways}&\vdots&\vdots\\
(-1)^{n-1}e_0&\cdots&\cdots&(-1)^{0}e_{n-1}
\end{bmatrix}
\cdot\begin{bmatrix}
\bar{\boldsymbol{x}}(\alpha_1)&\cdots&\bar{\boldsymbol{x}}(\alpha_n)
\end{bmatrix}\notag\\
=&\boldsymbol{T}
\begin{bmatrix}
\bar{\boldsymbol{x}}(\alpha_1)&\cdots&\bar{\boldsymbol{x}}(\alpha_n)
\end{bmatrix}\label{eq:conversion_ei_alpha_i}
\end{align}
where
\[\boldsymbol{T}=\begin{bmatrix}
&&&(-1)^{0}e_0\\
&&(-1)^{2}e_0&(-1)^{1}e_1\\
&\begin{sideways}
$\ddots$
\end{sideways}&\vdots&\vdots\\
(-1)^{2(n-1)}e_0&\cdots&\cdots&(-1)^{n-1}e_{n-1}
\end{bmatrix}.\]

\item The substitution of \eqref{eq:conversion_ei_alpha_i} into \eqref{eq:expansion_ei} immediately yields
\begin{align*}
&\boldsymbol{B_\omega}
\begin{bmatrix}
\bar{\boldsymbol{\omega}}(\alpha_1)&\cdots&\bar{\boldsymbol{\omega}}(\alpha_n)
\end{bmatrix}\\
=\,&
\begin{bmatrix}
\boldsymbol{B_\omega}\bar{\boldsymbol{\omega}}(\alpha_1)&\cdots&\boldsymbol{B_\omega}\bar{\boldsymbol{\omega}}(\alpha_n)
\end{bmatrix}\\
=\,&a_{n}\boldsymbol{U}^T
\boldsymbol{T}
\begin{bmatrix}
\bar{\boldsymbol{x}}(\alpha_1)&\cdots&\bar{\boldsymbol{x}}(\alpha_n)
\end{bmatrix}\begin{bmatrix}
G(\alpha_1)&&\\
&\ddots&\\
&&G(\alpha_n)
\end{bmatrix}\\
=\,&a_{n}\boldsymbol{U}^T
\boldsymbol{T}
\begin{bmatrix}
\bar{\boldsymbol{x}}(\alpha_1)G(\alpha_1)&\cdots&\bar{\boldsymbol{x}}(\alpha_n)G(\alpha_n)
\end{bmatrix}\\
=&a_{n}\boldsymbol{U}^T\boldsymbol{TU}
\begin{bmatrix}
\bar{\boldsymbol{\omega}}(\alpha_1)G(\alpha_1)&\cdots&\bar{\boldsymbol{\omega}}(\alpha_n)G(\alpha_n)
\end{bmatrix}.
\end{align*}

\end{enumerate}
\end{proof}

\begin{lemma}\label{lem:lowerpart}
Let $\boldsymbol{M}\in\mathbb{F}_{(n-k)\times n}$ with $k<n$ and
$\bar{\boldsymbol{\omega}}=(\omega_{n-1},$ $\ldots,\omega_1,\omega_{0})^{T}$ be a general basis of $\mathbb{F}_{n-1}[x]$.  Given
$\boldsymbol{\beta}=(\beta_1,\ldots,\beta_n)\in\overline{\mathbb{F}}^n$$ $, we have
\begin{equation}\label{eq:MXk}
\left|
\begin{array}{c}
\boldsymbol{M}\\
{\boldsymbol{X}_{\boldsymbol{\omega},k}}
\end{array}\right|=\frac{(-1)^k\left|
\begin{array}{ccc|c}
&\boldsymbol{M}\boldsymbol{W}(\boldsymbol{\beta})&&
\\\hline
\boldsymbol{\omega}_{k}(\beta_1)&\cdots&\boldsymbol{\omega}_k(\beta_n)&\boldsymbol{\omega}_k
\end{array}\right|}{|\boldsymbol{W}(\boldsymbol{\beta})|}
\end{equation}
where
$
{\boldsymbol{W}}(\boldsymbol{\beta})=
\setlength{\arraycolsep}{1pt}
\begin{bmatrix}
\bar{\boldsymbol{\omega}}(\beta_1)&\cdots&\bar{\boldsymbol{\omega}}(\beta_n)
\end{bmatrix}$ and $\boldsymbol{\omega}_k=(\omega_k,\ldots,\omega_1,\omega_0)^T$.
\end{lemma}

\begin{proof}
We start the proof by recalling the  following fact:\[\left|
\begin{array}{c}
\boldsymbol{M}\\
{\boldsymbol{X}_{\boldsymbol{\omega},k}}
\end{array}\right|\cdot|\boldsymbol{W}(\boldsymbol{\beta})|
=\left|
\begin{array}{c}
M\boldsymbol{W}(\boldsymbol{\beta})\\
{\boldsymbol{X}_{\boldsymbol{\omega},k}\boldsymbol{W}}(\boldsymbol{\beta})
\end{array}
\right|.
\]
Then we consider the matrix product in the lower part and carry out the following multiplication\begin{align*}
{\boldsymbol{X}_{\boldsymbol{\omega},k}\boldsymbol W}(\boldsymbol{\beta})=&
\setlength{\arraycolsep}{2.2pt}
\begin{bmatrix}
&&-1&&&\omega_k\\
&&&\ddots&&\vdots\\
&&&&-1&\omega_1
\end{bmatrix}
\begin{bmatrix}
\omega_{n-1}(\beta_1)&\cdots&\omega_{n-1}(\beta_n)\\
\vdots&\ddots&\vdots\\
\omega_{0}(\beta_1)
&\cdots&\omega_{0}(\beta_n)
\end{bmatrix}
\\
=&\begin{bmatrix}
\omega_k\omega_0(\beta_1)-\omega_k(\beta_1)&\cdots&\omega_k\omega_0(\beta_n)-\omega_k(\beta_n)\\
\vdots&\ddots&\vdots\\
\omega_1\omega_0(\beta_1)-\omega_1(\beta_1)&\cdots&\omega_1\omega_0(\beta_n)-\omega_1(\beta_n)
\end{bmatrix}.
\end{align*}
Taking determinants on both sides yields
\begin{align*}
&\left|
\begin{array}{c}
\boldsymbol{M}\\
{\boldsymbol{X}_{\boldsymbol{\omega},k}}
\end{array}\right|\cdot|\boldsymbol{W}(\boldsymbol{\beta})|\\
=&\left|
\begin{array}{ccc}
&\boldsymbol{M}\boldsymbol{W}(\boldsymbol{\beta})&
\\\hline
\omega_k\omega_0(\beta_1)-\omega_k(\beta_1)
&\hspace{-0.8em}\cdots
&\hspace{-0.8em}\omega_k\omega_0(\beta_n)-\omega_k(\beta_n)
\\
\vdots&\ddots&\vdots\\
\omega_1\omega_0(\beta_1)-\omega_1(\beta_1)&
\hspace{-0.8em}\cdots&
\hspace{-0.8em}\omega_1\omega_0(\beta_n)-\omega_1(\beta_n) \end{array}\right|\\
=&\left|
\begin{array}{ccc|c}
&\boldsymbol{M}\boldsymbol{W}(\boldsymbol{\beta})&&
\\\hline
\omega_k\omega_0(\beta_1)-\omega_k(\beta_1)
&\hspace{-0.8em}\cdots
&\hspace{-0.8em}\omega_k\omega_0(\beta_n)-\omega_k(\beta_n)&
\\
\vdots&\ddots&\vdots&\\
\omega_1\omega_0(\beta_1)-\omega_1(\beta_1)&
\hspace{-0.8em}\cdots&
\hspace{-0.8em}\omega_1\omega_0(\beta_n)-\omega_1(\beta_n) &\\
\omega_0(\beta_1)
&\hspace{-0.8em}\cdots
&\hspace{-0.8em}\omega_0(\beta_n)&~~1~~
\end{array}\right|.
\end{align*}
Subtracting the $(n+1)$-th row multiplied by $\omega_i$ from the $(n-i+1)$-th row for $i=1,\ldots,k$, we obtain
\[
\left|
\begin{array}{c}
\boldsymbol{M}\\
{\boldsymbol{X}_{\boldsymbol{\omega},k}}
\end{array}\right|\cdot|\boldsymbol{W}(\boldsymbol{\beta})|=\left|
\begin{array}{ccc|c}
&\boldsymbol{M}\boldsymbol{W}(\boldsymbol{\beta})&&
\\\hline
-\omega_k(\beta_1)
&\cdots
&-\omega_k(\beta_n)&-\omega_k
\\
\vdots&\ddots&\vdots&\vdots\\
-\omega_1(\beta_1)&
\cdots&
-\omega_1(\beta_n) &-\omega_1\\
\omega_0(\beta_1)
&\cdots
&\omega_0(\beta_n)&1
\end{array}\right|.
\]
By factoring out $-1$ from the last $k+1$ rows except for the last row and setting $1$ in the last column to be $\omega_0$,  the following is achieved:
\[
\left|
\begin{array}{c}
\boldsymbol{M}\\
{\boldsymbol{X}_{\boldsymbol{\omega},k}}
\end{array}\right|\cdot|\boldsymbol{W}(\boldsymbol{\beta})|=(-1)^k\left|
\setlength{\arraycolsep}{2.5pt}
\begin{array}{ccc|c}
&\boldsymbol{M}\boldsymbol{W}(\boldsymbol{\beta})&&
\\\hline
\boldsymbol{\omega}_{k}(\beta_1)&\cdots&\boldsymbol{\omega}_k(\beta_n)&\boldsymbol{\omega}_k
\end{array}\right|
\]
which is equivalent to what we want.
\end{proof}
\begin{remark}
The expression in the right side of \eqref{eq:MXk} should  be interpreted as the way for interpreting the expression of $S_k$ in Remark \ref{rem:interpreting}:
\begin{enumerate}[(1)]
\item Treat $\beta_1,\ldots,\beta_n$ as indeterminates and carry out the exact division, which results in a symmetric polynomial in terms of $\beta_1,\ldots,,\beta_{n}$.

\item Evaluate the polynomial with $\beta_1,\ldots,\beta_n$  assigned some specific values in $\bar{\mathbb{F}}$.
\end{enumerate}
Otherwise, the denominator will vanish when $\beta_i=\beta_j$ for some $i\ne j$.
\end{remark}

Now we are ready to prove  Lemma \ref{lem:root2coeffs}.

\begin{proof}
The proof is a bit long and will be divided into several steps.
\begin{enumerate}[(i)]
\item  Specializing $\boldsymbol{M}$ and $\boldsymbol{\beta}$ with
$\boldsymbol{B}_{\boldsymbol{\omega},k}$ and $\boldsymbol{\alpha}=(\alpha_1,\ldots,\alpha_n)$ where $\alpha_i$'s are the roots of $F$ over $\overline{\mathbb{F}}$, we have
\begin{equation}\label{eq:BX}
\left|
\begin{array}{c}
\boldsymbol{B}_{\boldsymbol{\omega},k}\\
{\boldsymbol{X}_{\boldsymbol{\omega},k}}
\end{array}\right|\cdot|\boldsymbol{W}(\boldsymbol{\alpha})|=(-1)^k\left|
\setlength{\arraycolsep}{1.5pt}
\begin{array}{ccc|c}
&\boldsymbol{B}_{\boldsymbol{\omega},k} \boldsymbol{W}(\boldsymbol{\alpha})&&
\\\hline
\boldsymbol{\omega}_{k}(\alpha_1)&\cdots&\boldsymbol{\omega}_k(\alpha_n)&\boldsymbol{\omega}_k
\end{array}\right|.
\end{equation}

\item Next we keep simplifying $\boldsymbol{B}_{\boldsymbol{\omega},k}\boldsymbol{W}(\boldsymbol{\alpha})$.
\begin{enumerate}[(a)]
\item  Since $\boldsymbol{B}_{\boldsymbol{\omega},k}=\begin{bmatrix}
\boldsymbol{I}_{n-k}&\boldsymbol{0}_{(n-k)\times k}\end{bmatrix}\boldsymbol{B}_{\boldsymbol{\omega}} $
where $\boldsymbol{I}_{n-k}$ is the identity matrix of order $n-k$, by Lemma \ref{lem:upperpart},
\begin{align*}
\boldsymbol{B}_{\boldsymbol{\omega},k}\boldsymbol{W}(\boldsymbol{\alpha})
=&\begin{bmatrix}
\boldsymbol{I}_{n-k}&\boldsymbol{0}_{(n-k)\times k}\end{bmatrix}\boldsymbol{B}_{\boldsymbol{\omega}}\boldsymbol{W}(\boldsymbol{\alpha})\\
=&a_{n}\begin{bmatrix}
\boldsymbol{I}_{n-k}&\boldsymbol{0}_{(n-k)\times k}\end{bmatrix}\boldsymbol{U}^T\boldsymbol{TU}\cdot
\begin{bmatrix}
\bar{\boldsymbol{\omega}}(\alpha_1)G(\alpha_1)&\cdots&\bar{\boldsymbol{\omega}}(\alpha_n)G(\alpha_n)
\end{bmatrix}.
\end{align*}

\item Note that $\boldsymbol{U}$ is unit upper triangular  and  thus $\boldsymbol{U}^T$ is unit lower triangular. Furthermore, $\boldsymbol{T}$ is unit reversed lower triangular. Thus we may partition them in the following way:
\begin{equation}\label{eq:UT}
\boldsymbol{U}=
\begin{bmatrix}
\boldsymbol{U}_1&*\\
&*
\end{bmatrix},\quad
\boldsymbol{U}^T=
\begin{bmatrix}
*&\\
*&\boldsymbol{U}_2
\end{bmatrix},\quad
\boldsymbol{T}=
\begin{bmatrix}
&\boldsymbol{T}_1\\
*&*
\end{bmatrix},
\end{equation}
where $\boldsymbol{U}_1,\boldsymbol{U}_2,\boldsymbol{T}_1\in\mathbb{F}^{(n-k)\times(n-k)}$. It follows that
\begin{align*}
&\begin{bmatrix}
\boldsymbol{I}_{n-k}&\boldsymbol{0}_{(n-k)\times k}\end{bmatrix}\boldsymbol{U}^T\boldsymbol{TU}\\
=&\begin{bmatrix}
\boldsymbol{I}_{n-k}&\boldsymbol{0}_{(n-k)\times k}
\end{bmatrix}
\begin{bmatrix}
*&\\
*&\boldsymbol{U}_2
\end{bmatrix}
\begin{bmatrix}
&\boldsymbol{T}_1\\
*&*
\end{bmatrix}
\begin{bmatrix}
\boldsymbol{U}_1&*\\
&*
\end{bmatrix}\\
=&\begin{bmatrix}
\boldsymbol{0}_{(n-k)\times k}&\boldsymbol{U}_2\boldsymbol{T}_1\boldsymbol{U}_1
\end{bmatrix}.
\end{align*}

\item
Therefore,
\begin{align*}
\boldsymbol{B}_{\boldsymbol{\omega},k}\boldsymbol{\omega}(\boldsymbol{\alpha})=&
a_n\begin{bmatrix}
\boldsymbol{0}_{(n-k)\times k}&\boldsymbol{U}_2\boldsymbol{T}_1\boldsymbol{U}_1
\end{bmatrix}\cdot
\begin{bmatrix}
\bar{\boldsymbol{\omega}}(\alpha_1)G(\alpha_1)&\cdots&\bar{\boldsymbol{\omega}}(\alpha_n)G(\alpha_n)
\end{bmatrix}.
\end{align*}

\item Now we partition the matrix
$$
\begin{bmatrix}
\bar{\boldsymbol{\omega}}(\alpha_1)G(\alpha_1)~\cdots~\bar{\boldsymbol{\omega}}(\alpha_n)G(\alpha_n)
\end{bmatrix}$$
into two blocks with the upper block consisting of $k$ rows and the lower block consisting of $(n-k)$ rows. After carrying out the matrix multiplication, we obtain
\begin{align*}
\boldsymbol{B}_{\boldsymbol{\omega},k}\boldsymbol{W}(\boldsymbol{\alpha})
=&\,a_{n}  \boldsymbol{U}_2\boldsymbol{T}_1\boldsymbol{U}_1
\,\cdot
\setlength{\arraycolsep}{.5pt}
\begin{bmatrix}
\boldsymbol{\omega}_{n-k-1}(\alpha_1)G(\alpha_1) & \cdots & \boldsymbol{\omega}_{n-k-1}(\alpha_n)G(\alpha_n)
\end{bmatrix}
\end{align*}

\item Therefore,
\begin{align*}
&\left[
\begin{array}{ccc|c}
&\boldsymbol{B}_{\boldsymbol{\omega},k} \boldsymbol{W}(\boldsymbol{\alpha})&&
\\\hline
\boldsymbol{\omega}_{k}(\alpha_1)&\cdots&\boldsymbol{\omega}_k(\alpha_n)&\boldsymbol{\omega}_k
\end{array}\right]\\
=&\begin{bmatrix}
a_{n}  \boldsymbol{U}_2\boldsymbol{T}_1\boldsymbol{U}_1&\\
&\boldsymbol{I}_{k+1}
\end{bmatrix}\cdot
\left[
\begin{array}{ccc|c}
\boldsymbol{\omega}_{n-k-1}(\alpha_1)G(\alpha_1) & \cdots & \boldsymbol{\omega}_{n-k-1}(\alpha_n)G(\alpha_n)&
\\\hline
\boldsymbol{\omega}_{k}(\alpha_1)&\cdots&\boldsymbol{\omega}_k(\alpha_n)&\boldsymbol{\omega}_k
\end{array}\right]
\end{align*}
where $\boldsymbol{I}_{k+1}$ is the identity matrix of order $k+1$.
\end{enumerate}

\item The current step is devoted to calculate the determinant of the above matrix which is closely related with $S_k$ in Lemma
\ref{lem:std2nstd}.
\begin{enumerate}[(a)]
\item  Recall Lemma \ref{lem:std2nstd} and convert it into

\begin{align*}
S_k={c\cdot\left|
\setlength{\arraycolsep}{1.0pt}
\begin{array}{lclc}
\boldsymbol{\omega}_{n-k-1}(\alpha_1)G(\alpha_1) & \cdots & \boldsymbol{\omega}_{n-k-1}(\alpha_n)G(\alpha_n) &\\
\boldsymbol{\omega}_{k}(\alpha_1) & \cdots & \boldsymbol{\omega}_{k}(\alpha_n)&\boldsymbol{\omega}_{k}
\end{array}\right|}{
}\big/|\boldsymbol{W}(\boldsymbol{\alpha})|
\end{align*}

\item By (ii.e),
\[
\left|
\setlength{\arraycolsep}{1.0pt}
\begin{array}{lclc}
\boldsymbol{\omega}_{n-k-1}(\alpha_1)G(\alpha_1) & \cdots & \boldsymbol{\omega}_{n-k-1}(\alpha_n)G(\alpha_n) &\\
\boldsymbol{\omega}_{k}(\alpha_1) & \cdots & \boldsymbol{\omega}_{k}(\alpha_n)&\boldsymbol{\omega}_{k}
\end{array}\right|=\frac{\left|
\begin{array}{ccc|c}
&\boldsymbol{B}_{\boldsymbol{\omega},k} \boldsymbol{W}(\boldsymbol{\alpha})&&
\\\hline
\boldsymbol{\omega}_{k}(\alpha_1)&\cdots&\boldsymbol{\omega}_k(\alpha_n)&\boldsymbol{\omega}_k
\end{array}\right|}{\begin{vmatrix}
a_{n}  \boldsymbol{U}_2\boldsymbol{T}_1\boldsymbol{U}_1&\\
&\boldsymbol{I}_{k+1}
\end{vmatrix}}
\]

\item With $M$ and $\boldsymbol{\beta}$ in Lemma \ref{lem:lowerpart} specialized with $\boldsymbol{B}_{\boldsymbol{\omega},k}$ and $\boldsymbol{\alpha}$, we have
\[\left|
\begin{array}{ccc|c}
&\boldsymbol{B}_{\boldsymbol{\omega},k}\boldsymbol{W}(\boldsymbol{\alpha})&&
\\\hline
\boldsymbol{\omega}_{k}(\alpha_1)&\cdots&\boldsymbol{\omega}_k(\alpha_n)&\boldsymbol{\omega}_k
\end{array}\right|
=(-1)^k\left|
\begin{array}{c}
\boldsymbol{B}_{\boldsymbol{\omega},k}\\
{\boldsymbol{X}_{\boldsymbol{\omega},k}}
\end{array}\right||\boldsymbol{W}(\boldsymbol{\alpha})|.
\]

\item Combining (iii.a)-(iii.c), one may obtain
\[S_k=\frac{c}{(-1)^k|a_n\boldsymbol{U}_2\boldsymbol{T}_1\boldsymbol{U}_1|}
\cdot\left|
\begin{array}{c}
\boldsymbol{B}_{\boldsymbol{\omega},k} \\
{\boldsymbol{X}_{\boldsymbol{\omega},k}}
\end{array}\right|.\]
\end{enumerate}

\item  In the last step, we will figure out what the coefficient in front of the resulting determinant is above.

\begin{enumerate}[(a)]
\item Recall $\boldsymbol{U}$ is unit upper triangular and $\boldsymbol{T}$ is unit reversed upper triangular. Thus $\boldsymbol{U}_1$ and $\boldsymbol{U}_2$ are both unit upper triangular and
$\boldsymbol{T}_1$ is unit reversed lower triangular. Moreover, the orders of the three matrices are all $n-k$. Therefore,
\[
|a_n\boldsymbol{U}_2\boldsymbol{T}_1\boldsymbol{U}_1|=a_n^{n-k}|\boldsymbol{U}_2|\cdot|\boldsymbol{T}_1|\cdot |\boldsymbol{U}_1|=a_n^{n-k}\cdot(-1)^{\sum_{i=1}^{n-k}(i+1)}=(-1)^{{n-k\choose 2}}\cdot a_n^{n-k},
\]
which indicates that
\[S_k=(-1)^{k+{n-k\choose 2}}ca_{n}^{k-n  }\left|
\begin{array}{c}
\boldsymbol{B}_{\boldsymbol{\omega},k} \\
{\boldsymbol{X}_{\boldsymbol{\omega},k}}
\end{array}\right|.\]

\item
We keep simplify the constant factor in the above equation below.
\begin{align*}
&\,(-1)^{k+{n-k\choose 2}}ca_{n}^{k-n  }
=\,(-1)^{k+{n-k\choose 2}}\cdot(-1)^{k}a_{n}^{m-k}a_{n}^{k-n  }
=(-1)^{n-k\choose 2}a_n^{m-n}.
\end{align*}

\item To sum up, we have
\[S_k=(-1)^{n-k\choose 2}a_n^{m-n}\cdot\left|
\begin{array}{c}
\boldsymbol{B}_{\boldsymbol{\omega},k}\\
{\boldsymbol{X}_{\boldsymbol{\omega},k}}
\end{array}\right|.
\]
\end{enumerate}
\end{enumerate}
\end{proof}

\subsection{Converting $S_k$ in general basis from a single determinant to a determinental polynomial}

In this subsection, we prove a more general result than what is needed. The more general result is presented in the hope that it would be useful in some other occasions.

\begin{lemma}\label{lem:det2detp}
Given $\boldsymbol{\omega}=(\omega_{n},\ldots,\omega_{1},\omega_0)^{T}$  which is a {general basis} of $\mathbb{F}_{n}[x]$ and $\boldsymbol{M}\in\mathbb{F}^{(n-k)\times n}$ where $k<n$, let
$\boldsymbol{X}_{\boldsymbol{\omega},k}$ be as in \eqref{eq:Xw}.
Then
\[\operatorname*{detp}\nolimits_{\boldsymbol{\omega}}\boldsymbol{M}
=\begin{vmatrix}
\boldsymbol{M}\\
\boldsymbol{X}_{\boldsymbol{\omega},k}
\end{vmatrix}.\]
\end{lemma}

\begin{proof}
Denote the $i$-th column of $\boldsymbol{M}$ with $\boldsymbol{M}_i$. Then
\[
\begin{vmatrix}
\boldsymbol{M}\\
\boldsymbol{X}_{\boldsymbol{\omega},k}
\end{vmatrix}=
\begin{vmatrix}
\boldsymbol{M}_{1}&\cdots&\boldsymbol{M}_{n-k}&\cdots&\cdots&\boldsymbol{M}_n\\
&&-1&&&\omega_k\\
&&&\ddots&&\vdots\\
&&&&-1&\omega_1
\end{vmatrix}
\]
Adding the $(n-i)$-th column multiplied by $\omega_i$ for $i=1,\ldots,k$ to the last column yields
\[
\begin{vmatrix}
\boldsymbol{M}\\
\boldsymbol{X}_{\boldsymbol{\omega},k}
\end{vmatrix}=
\begin{vmatrix}
\boldsymbol{M}_{1}&\cdots&\boldsymbol{M}_{n-k}&
\cdots&\cdots&\sum_{i=1}^k{\boldsymbol{M}_{n-i}\omega_i}+\boldsymbol{M}_n\\
&&-1&&&0\\
&&&\ddots&&\vdots\\
&&&&-1&0
\end{vmatrix}.
\]
Note that there is only one non-zero entry in the $(n-i)$-th row of the determinant in the right-hand-side for $i=0,1\ldots,k-1$ and $\omega_0=1$. Thus its expansion  results in
\[
\begin{vmatrix}
\boldsymbol{M}\\
\boldsymbol{X}_{\boldsymbol{\omega},k}
\end{vmatrix}=
\begin{vmatrix}
\boldsymbol{M}_{1}&\cdots&\boldsymbol{M}_{n-k-1}&\sum_{i=0}^k{\boldsymbol{M}_{n-i}\omega_i}
\end{vmatrix}.\]
Keep simplifying the determinant and we get
\begin{align*}
\begin{vmatrix}
\boldsymbol{M}\\
\boldsymbol{X}_{\boldsymbol{\omega},k}
\end{vmatrix}=&
\sum_{i=0}^k\begin{vmatrix}
\boldsymbol{M}_{1}&\cdots&\boldsymbol{M}_{n-k-1}&{\boldsymbol{M}_{n-i}\omega_i}
\end{vmatrix}\\
=&\sum_{i=0}^k\begin{vmatrix}
\boldsymbol{M}_{1}&\cdots&\boldsymbol{M}_{n-k-1}&{\boldsymbol{M}_{n-i}}
\end{vmatrix}{\omega_i}\\
=&
\sum_{i=0}^k\begin{vmatrix}
\hat{\boldsymbol{M}}_{i}
\end{vmatrix}{\omega_i}\\
=&\operatorname*{detp}\nolimits_{\boldsymbol{\omega}}\boldsymbol{M},
\end{align*}
which is equivalent with the result we want.
\end{proof}

After specializing $\boldsymbol{M}$ in Lemma \ref{lem:det2detp} with $\boldsymbol{B}_{\boldsymbol{\omega},k}$ in Lemma \ref{lem:root2coeffs}, we immediately deduce the following result.

\begin{corollary}
Given a general basis $\boldsymbol{\omega}$ of $\mathbb{F}_n[x]$ and $F$, $G$ as in \eqref{eq:F+G},
let $\boldsymbol{B}_{\boldsymbol{\omega},k}$ be the submatrix of  the B\'ezout matrix of $F$ and $G$ with respect to $x$ in the basis $\boldsymbol{\omega}$ by deleting the last $k$ rows and
$\boldsymbol{X}_{\boldsymbol{\omega},k}$ be as in \eqref{eq:Xw}.
Then
\begin{equation}\label{eq:det2detp}
\begin{vmatrix}
\boldsymbol{B}_{\boldsymbol{\omega},k}\\
\boldsymbol{X}_{\boldsymbol{\omega},k}
\end{vmatrix}=\operatorname*{detp}\nolimits_{\boldsymbol{\omega}}
\boldsymbol{B}_{\boldsymbol{\omega},k}.
\end{equation}
\end{corollary}

\subsection{Proof of Theorem \ref{thm:main}}
Now we are ready to prove the main theorem (i.e., Theorem \ref{thm:main}).

\begin{proof}
By Lemma \ref{lem:root2coeffs}, we have
\[S_k=c_{\boldsymbol{\omega}}\cdot
\begin{vmatrix}
\boldsymbol{B}_{\boldsymbol{\omega},k}\\
\boldsymbol{X}_{\boldsymbol{\omega},k}
\end{vmatrix}.\]
By integrating \eqref{eq:det2detp} into the above equation, we obtain
\[
S_k=c_{\boldsymbol{\omega}}\cdot\operatorname*{detp}\nolimits_{\boldsymbol{\omega}}\boldsymbol{B}_{\boldsymbol{\omega},k}.\]
where $c_{\boldsymbol{\omega}}=(-1)^{n-k\choose 2}a_n^{m-n}$,
which completes the proof.
\end{proof}

\section{Conclusion and Perspectives}
\label{sec:conclusion}

In this paper, we propose an approach to formulate the B\'ezout-type subresultant polynomial of univariate polynomials expressed in general basis with the help of B\'ezout matrix in non-standard basis. Although the basis is changed, the essential properties of subresultant are maintained. This study is motivated by the observation that the  B\'ezout matrix of polynomials in general basis is usually simpler than the  B\'ezout matrix of polynomials obtained by converting the polynomials into standard basis (see Example \ref{ex}), which is also true for the resulting subresultant  polynomials.
The simple expression of the B\'ezout matrix in general basis could be very helpful for exploring the nice hidden structure of subresultants in general basis. However, in order to come up with an efficient way to compute resultant matrices and subresultant  polynomials in general basis, it is necessary to exploit the structure of given basis, which will become a critical issue to be studied in the next stage.

Another related question is how to design algorithms for computing subresultant  polynomials of polynomials in other bases  which share similar properties as general basis (e.g., Bernstein basis \cite{bl2004}). It is also worthy of further investigation in the future.

\medskip\noindent\textbf{Acknowledgements.} This research was supported by National Natural Science Foundation of China under Grant Nos. 12261010
and 11801101, and the Natural Science Cultivation Project of Guangxi Minzu University under Grant No. 2022MDKJ001.

%\bibliographystyle{plain}
%\bibliography{ref}
\end{document}